\documentclass[letterpaper, 10 pt, conference]{IEEEtran}  % Comment this line out

\IEEEoverridecommandlockouts

\let\proof\relax 
\let\endproof\relax 

\usepackage{bbm}
\usepackage{mathptmx}
\usepackage{graphicx} % more modern
\usepackage{amssymb}
\usepackage{url}
\usepackage{subfigure}
\usepackage{amsmath}
\usepackage{amsthm,amsfonts,bm}
\usepackage{soul}
\usepackage{cite}
\usepackage{hyperref}
\usepackage{algorithm,algpseudocode}

\usepackage{tikz}
\usetikzlibrary{arrows}

\newtheorem{definition}{Definition}

\newtheorem{assumption}{Assumption}

\newtheorem{proposition}{Proposition}

\newtheorem{theorem}{Theorem}

\begin{document}

\title{\huge{Differentially Private ADMM-Based Distributed Discrete Optimal Transport for Resource Allocation}
}

\author{Jason Hughes and Juntao Chen
\thanks{The authors are with the Department of Computer and Information Sciences, Fordham University, New York, NY, 10023 USA. E-mail: \{jhughes50,jchen504\}@fordham.edu}
\thanks{This work was supported in part by the National Science Foundation under Grant ECCS-2138956, and in part by a Faculty Research Grant from Fordham Office of Research.}}

\maketitle

\begin{abstract}
Optimal transport (OT) is a framework that can guide the design of efficient resource allocation strategies in a network of multiple sources and targets. To ease the computational complexity of large-scale transport design, we first develop a distributed algorithm based on the alternating direction method of multipliers (ADMM). However, such a distributed algorithm is vulnerable to sensitive information leakage when an attacker intercepts the transport decisions communicated between nodes during the distributed ADMM updates. To this end, we propose a privacy-preserving distributed mechanism based on output variable perturbation by adding appropriate randomness to each node's decision before it is shared with other corresponding nodes at each update instance. We show that the developed scheme is differentially private, which prevents the adversary from inferring the node's confidential information even knowing the transport decisions. Finally, we corroborate the effectiveness of the devised algorithm through case studies.

\end{abstract}

%\begin{keywords}
%Discrete Optimal Transport, Distributed Algorithm, Differential Privacy, Resource Allocation
%\end{keywords}

\section{Introduction}
The optimal transport (OT) paradigm can be leveraged to guide the most efficient allocation of a limited amount of resources from a set of sources to a set of targets by considering their heterogeneous preferences \cite{jhughes2021fair,zhang2019consensus}. The standard OT framework computes the transport strategy in a centralized manner, which requires the source and target nodes to send their information to a centralized transport planner. This centralized computation mechanism is not scalable when the transport network includes a large number of participants. 
Thus, it is imperative to design a computationally efficient scheme that applies to large-scale transport design. 

To this end, distributed algorithm based on the alternating direction method of multipliers (ADMM) can be used to achieve this goal. In the distributed computation scheme, each node communicates directly with the connected nodes regarding the transport decisions and reaches a consensus through iterative negotiations. Under this paradigm, the central planner does not necessarily need to coordinate the resource matching. The distributed OT design eliminates the necessity of a centralized communication network where each node reports their preference information to the central planner. Instead, the communication occurs between each pair of connected source and target nodes enabled by a peer-to-peer network. Thus, the ADMM-based distributed algorithm does not require sharing all the nodes' information over the network.

However, the distributed OT algorithm still faces adversarial threats \cite{hughes2021deceptive}. Specifically, the nodes need to communicate their computed resource transport preferences with the connected nodes at each update step in the algorithm.  This information could be intercepted by an adversary during its transmission over the communication network (e.g., through eavesdropping attack). The attacker can then use it to infer the private information at each participating node (e.g., node's utility parameters used for the design of transport plan). 

The privacy concerns of the distributed OT motivate us to develop an efficient privacy-preserving mechanism that can protect the nodes' sensitive utility information. To do this, we resort to the powerful differential privacy technique \cite{dwork2014_book}. Specifically, we develop an output variable perturbation-based differentially private distributed OT scheme. In this algorithm, instead of sharing the authentic transport strategies directly between connected source and target nodes, each node perturbs their transport decisions by adding a random noise drawn from an appropriate distribution with specified parameters at each step. The proposed algorithm prevents leakage of sensitive information of participants in the network even if the transport strategies shared between nodes during updates are captured by the adversary.

The contributions of this paper are presented as follows.
\begin{enumerate}
    \item We develop a distributed OT design framework based on the  alternating direction method of multipliers to compute the OT strategies efficiently.
    \item We incorporate privacy consideration into the distributed OT and propose a differentially private distributed OT algorithm based on an output variable perturbation mechanism.
    \item We demonstrate the effectiveness of the developed algorithm through case studies and characterize the trade-off between a node's privacy and transport utility.
\end{enumerate}

\textit{Related Works.} Differential privacy has been applied to many fields, especially in artificial intelligence and machine learning. Differential privacy has been studied with specific application to ADMM-based distributed algorithms for both learning and optimization in \cite{2020HuangDPADMM,czhang2019ADMMopt}. Specifically, perturbation-based ADMM algorithms were developed to improve privacy in classification learning problems  \cite{zhang2017privateADMMlearning,zhang2018improving}. Differential privacy has also been leveraged to investigate privacy issues in empirical risk minimization \cite{chaudhuri2011privateERM}, support vector machines \cite{zhang2019privateSVM} and deep learning \cite{Abadi2016privateDL}. Additionally, differential privacy has been applied to tackle the privacy issues in various societal applications, including fog computing \cite{du2017differential}, private data trading through contracts \cite{khalili2021designing}, federated learning in the Internet of things \cite{zhao2020local}, and vehicular networks \cite{zhang2018distributed}. In this work, we address the privacy concerns in the ADMM-based distributed OT algorithm based on differential privacy and develop a protection scheme that has a theoretical guarantee to maintain the privacy of the information at each participating transport node. 

%The rest of this paper is organized as follows. Section \ref{sec:DDOT} presents the basics of discrete optimal transport over a network and develops a distributed algorithm to compute the solution. Section \ref{sec:DPOT} concerns the privacy of the OT framework and proposes a differentially private distributed OT algorithm. Section \ref{sec:case} presents case studies to illustrate the results, and Section \ref{sec:conclusion} concludes the paper.

\section{Discrete Optimal Transport over Networks and Distributed Algorithm}\label{sec:DDOT}

This section presents the framework of discrete optimal transport over a network and then develops a distributed algorithm to compute the optimal transport plan.

\subsection{Discrete Optimal Transport}
We denote by $\mathcal{X}:=\{1, ..., |\mathcal{X}|\}$ a set of destination/target nodes that receive the resources, and $\mathcal{Y}:=\{|\mathcal{X}|+1, ..., |\mathcal{X}|+|\mathcal{Y}|\}$ a set of origin/source nodes that distribute resources to the targets over a transport network. Additionally, we define $\mathcal{P} = \mathcal{X} \cup \mathcal{Y}$ as the set of all nodes. Each source node $y\in\mathcal{Y}$ is connected to a number of target nodes denoted by $\mathcal{X}_y$, representing that $y$ can choose to allocate its resources to a specific group of destinations $\mathcal{X}_y$. Similarly, each target node $x\in\mathcal{X}$ can receive resources from multiple source nodes, and this set of resource suppliers to target $x$ is denoted by $\mathcal{Y}_x$. It can be seen that the resources are transported over a bipartite network, where one side of the network consists of all source nodes and the other includes all destination nodes. This bipartite network is not necessarily complete because of constrained matching policies between participants. We further denote by $\mathcal{E}$ the set of all feasible transport paths in the network, i.e., $\mathcal{E}:=\{\{x,y\}|x\in\mathcal{X}_y,y\in\mathcal{Y}\}$. Here, $\mathcal{E}$ also refers to the set of all edges in the established bipartite graph for resource transportation.

We next denote by $\pi_{xy}\in\mathbb{R}_+$ the amount of resources transported from the origin node $y\in\mathcal{Y}$ to the destination node $x\in\mathcal{X}$, where $\mathbb{R}_+$ is the set of nonnegative real numbers. Let $\Pi:=\{\pi_{xy}\}_{x\in\mathcal{X}_y,y\in\mathcal{Y}}$ be the designed transport plan. Then, the centralized optimal transport problem can be formulated as follows:
\begin{equation}\label{OT1:eqn}
\begin{aligned}
    \max_{\Pi}\ \sum_{x\in\mathcal{X}} \sum_{y\in\mathcal{Y}_x}& t_{xy}(\pi_{xy}) + \sum_{y\in\mathcal{Y}} \sum_{x\in\mathcal{X}_y} s_{xy}(\pi_{xy})\\
    \mathrm{s.t.}\quad &\underline{p}_{x}\leq \sum_{y\in\mathcal{Y}_x} \pi_{xy}\leq \bar{p}_{x},\ \forall x\in\mathcal{X},\\
    &\underline{q}_{y}\leq \sum_{x\in\mathcal{X}_y} \pi_{xy}\leq \bar{q}_{y},\ \forall y\in\mathcal{Y},\\
    &\pi_{xy}\geq 0,\ \forall \{x,y\} \in\mathcal{E},
\end{aligned}
\end{equation}
where $t_{xy}:\mathbb{R}_+\rightarrow\mathbb{R}$ and $s_{xy}:\mathbb{R}_+\rightarrow\mathbb{R}$ are utility functions for target node $x$ and source node $y$, respectively. Furthermore, $\bar{p}_x\geq \underline{p}_{x}\geq 0$, $\forall x\in\mathcal{X}$ and $\bar{q}_y\geq \underline{q}_{y}\geq 0$, $\forall y\in\mathcal{Y}$. The constraints $\underline{p}_{x}\leq \sum_{y\in\mathcal{Y}_x} \pi_{xy}\leq \bar{p}_{x}$ and $\underline{q}_{y}\leq \sum_{x\in\mathcal{X}_y} \pi_{xy}\leq \bar{q}_{y}$ capture the limitations on the amount of requested and transferred resources at the target $x$ and source $y$, respectively. 

We have the following assumption on the utility functions.
\begin{assumption}\label{assump:1}
The utility functions $t_{xy}$ and $s_{xy}$ are concave and monotonically increasing on $\pi_{xy}$, $\forall x\in\mathcal{X},\forall y\in\mathcal{Y}$. Moreover, they are continuously differentiable with $t'_{xy} \leq \rho$ and $s'_{xy} \leq \rho$, where $\rho$ is a positive constant.
\end{assumption}

A rich class of functions satisfy the conditions in Assumption \ref{assump:1}. For example, the utility functions $t_{xy}$ and $s_{xy}$ can be linear on $\pi_{xy}$, indicating a linear growth of benefits on the amount of transferred and consumed resources. %These two functions can also admit a logarithmic form, capturing that the marginal utility decreases as the amount of transported resources increase.

\subsection{Distributed Optimal Transport}\label{sec:dis_alg}
Next, we establish a distributed algorithm for computing the optimal transport strategy in \eqref{OT1:eqn}. Our first step is to reformulate the optimization problem by introducing ancillary variables $\pi_{xy,t}$ and $\pi_{xy,s}$. The additional subscripts $t$ and $s$ indicate that the corresponding parameters belong to the target node or the source node, respectively. We then set $\pi_{xy} = \pi_{xy,t}$ and $\pi_{xy} = \pi_{xy,s}$, indicating that the solutions proposed by the targets and sources are consistent. This reformulation facilitates the design of a distributed algorithm which allows us to iterate through the process in obtaining the optimal transport plan. To this end, the reformulated optimal transport problem is presented as follows:
\begin{equation}\label{OT2:eqn}
\begin{aligned}
\min_{\Pi_t \in \mathcal{F}_t, \Pi_s \in \mathcal{F}_s,\Pi} & -\sum_{x\in\mathcal{X}} \sum_{y\in\mathcal{Y}_x} t_{xy}(\pi_{xy,t}) - \sum_{y\in\mathcal{Y}} \sum_{x\in\mathcal{X}_y} s_{xy}(\pi_{xy,s}) \\
\mathrm{s.t.}\quad & \pi_{xy,t} = \pi_{xy},\ \forall \{x,y\}\in\mathcal{E},\\
& \pi_{xy,s} = \pi_{xy},\ \forall \{x,y\}\in\mathcal{E},
\end{aligned}
\end{equation}
where $\Pi_t:=\{\pi_{xy,t}\}_{x\in\mathcal{X}_y,y\in\mathcal{Y}}$, $\Pi_s:=\{\pi_{xy,s}\}_{x\in\mathcal{X},y\in\mathcal{Y}_x}$, $\mathcal{F}_t := \{ \Pi_t | \pi_{xy,t} \geq 0, \underline{p}_x \leq \sum_{y \in \mathcal{Y}_x} \pi_{xy,t} \leq \bar{p}_x,\ \{x,y\} \in \mathcal{E}\}$, and $\mathcal{F}_s := \{ \Pi_s | \pi_{xy,s} \geq 0, \underline{q}_y \leq \sum_{x \in \mathcal{X}_y} \pi_{xy,s} \leq \bar{q}_y,\ \{x,y\} \in \mathcal{E} \}$.

We resort to the alternating direction method of multipliers (ADMM) \cite{boyd2011distributed} to develop a distributed computational algorithm. First, let $\alpha_{xy,s}$ and $\alpha_{xy,t}$ be the Lagrangian multipliers associated with the constraint $\pi_{xy,s} = \pi_{xy}$ and $\pi_{xy,t} = \pi_{xy}$, respectively. The Lagrangian function associated with the optimization problem \eqref{OT2:eqn} can then be written as follows:
\begin{equation}\label{Lag:eqn}
\begin{split}
    &L \left(\Pi_{t}, \Pi_{s}, \Pi, \alpha_{xy,t}, \alpha_{xy,s} \right) =- \sum_{x\in\mathcal{X}} \sum_{y\in\mathcal{Y}_x} t_{xy}(\pi_{xy,t})\\
    & - \sum_{y\in\mathcal{Y}} \sum_{x\in\mathcal{X}_y} s_{xy}(\pi_{xy,s})  +\sum_{x\in\mathcal{X}} \sum_{y\in\mathcal{Y}_x} \alpha_{xy,t} (\pi_{xy,t} - \pi_{xy}) \\ &+\sum_{y\in\mathcal{Y}} \sum_{x\in\mathcal{X}_y} \alpha_{xy,s} (\pi_{xy} - \pi_{xy,s}) + \frac{\eta}{2} \sum_{x\in\mathcal{X}} \sum_{y\in\mathcal{Y}_x} (\pi_{xy,t} - \pi_{xy})^2 \\
    &+ \frac{\eta}{2} \sum_{y\in\mathcal{Y}} \sum_{x\in\mathcal{X}_y} (\pi_{xy} - \pi_{xy,s})^2,
\end{split}
\end{equation}
where $\eta > 0$ is a positive scalar constant controlling the convergence rate in the algorithm designed below.

Note that in \eqref{Lag:eqn}, the last two terms $\frac{\eta}{2} \sum_{x\in\mathcal{X}} \sum_{y\in\mathcal{Y}_x} (\pi_{xy,t} - \pi_{xy})^2$ and $\frac{\eta}{2} \sum_{y\in\mathcal{Y}} \sum_{x\in\mathcal{X}_y} (\pi_{xy} - \pi_{xy,s})^2$, acting as penalization, are quadratic. Hence, the Lagrangian function $L$ is strictly convex, ensuring the existence of a unique optimal solution. 

We next apply ADMM to the minimization problem in \eqref{OT2:eqn}. The designed distributed algorithm is presented in the following proposition.

\begin{proposition}\label{prop:1}
 The iterative steps of applying ADMM to problem \eqref{OT2:eqn} are summarized as follows:
\begin{equation}\label{ADMM1_eqn1}
\begin{split}
    \Pi_{x,t}(k+1) &\in \arg \min_{\Pi_{x,t}\in\mathcal{F}_{x,t}} - \sum_{y\in\mathcal{Y}_x} t_{xy}(\pi_{xy,t}) \\ 
    &+ \sum_{y\in\mathcal{Y}_x} \alpha_{xy,t}(k) \pi_{xy,t} + \frac{\eta}{2} \sum_{y\in\mathcal{Y}_x} (\pi_{xy,t} - \pi_{xy}(k))^2,
\end{split}
\end{equation}
\begin{equation}\label{ADMM1_eqn2}
\begin{aligned}
        \Pi_{y,s}(k+&1) \in \arg \min_{\Pi_{y,s}\in\mathcal{F}_{y,s}} - \sum_{x\in\mathcal{X}_y} s_{xy}(\pi_{xy,s})  \\ &-\sum_{x\in\mathcal{X}_y} \alpha_{xy,s}(k)\pi_{xy,s} + \frac{\eta}{2} \sum_{x\in\mathcal{X}_y} (\pi_{xy}(k) - \pi_{xy,s})^2,
\end{aligned}
\end{equation}
\begin{equation}\label{ADMM1_eqn3}
\begin{split}
    \pi_{xy}(&k+1)= \arg \min_{\pi_{xy}} - \alpha_{xy,t}(k)\pi_{xy} + \alpha_{xy,s}(k)\pi_{xy} \\
    &+\frac{\eta}{2}(\pi_{xy,t}(k+1) - \pi_{xy})^2 + \frac{\eta}{2}(\pi_{xy} - \pi_{xy,s}(k+1))^2,
\end{split}
\end{equation}
\begin{equation}\label{ADMM1_eqn4}
\begin{split}
    \alpha_{xy,t}(k+1) = \alpha_{xy,t}(k) + \eta(\pi_{xy,t}(k+1)-\pi_{xy}(k+1))^2,
\end{split}
\end{equation}
\begin{equation}\label{ADMM1_eqn5}
\begin{split}
    \alpha_{xy,s}(k+1) = \alpha_{xy,s}(k) + \eta(\pi_{xy}(k+1)-\pi_{xy,s}(k+1))^2,
\end{split}
\end{equation}
where $\Pi_{\tilde{x},t}:=\{\pi_{xy,t}\}_{y\in\mathcal{Y}_x,x=\tilde{x}}$ represents the solution at target node $\tilde{x}\in\mathcal{X}$, and $\Pi_{\tilde{y},s}:=\{\pi_{xy,s}\}_{x\in\mathcal{X}_y,y=\tilde{y}}$ represents the proposed solution at source node $\tilde{y}\in\mathcal{Y}$. In addition, $\mathcal{F}_{x,t} := \{ \Pi_{x,t} | \pi_{xy,t} \geq 0, y\in\mathcal{Y}_x, \underline{p}_x \leq \sum_{y \in \mathcal{Y}_x} \pi_{xy,t} \leq \bar{p}_x\}$, and $\mathcal{F}_{y,s} := \{ \Pi_{y,s} | \pi_{xy,s} \geq 0, x\in\mathcal{X}_y, \underline{q}_y \leq \sum_{x \in \mathcal{X}_y} \pi_{xy,s} \leq \bar{q}_y\}$.
\end{proposition}

\begin{proof}
See Appendix \ref{proof_prop:1}.
\end{proof}

We can simplify steps \eqref{ADMM1_eqn1}-\eqref{ADMM1_eqn5} down to four steps, and the results are summarized below.

\begin{proposition}\label{prop:2}
The iterations \eqref{ADMM1_eqn1}-\eqref{ADMM1_eqn5} can be simplified as
\begin{equation}\label{ADMM2_eqn1}
\begin{split}
    \Pi_{x,t}(k+1) &\in \arg \min_{\Pi_{x,t}\in\mathcal{F}_{x,t}} - \sum_{y\in\mathcal{Y}_x} t_{xy}(\pi_{xy,t}) \\
   & + \sum_{y\in\mathcal{Y}_x} \alpha_{xy}(k) \pi_{xy,t} + \frac{\eta}{2} \sum_{y\in\mathcal{Y}_x} \left(\pi_{xy,t} - \pi_{xy}(k)\right)^2,
\end{split}
\end{equation}
\begin{equation}\label{ADMM2_eqn2}
\begin{split}
    \Pi_{y,s}(k+&1) \in \arg \min_{\Pi_{y,s}\in\mathcal{F}_{y,s}} - \sum_{x\in\mathcal{X}_y} s_{xy}(\pi_{xy,s}) \\ &-\sum_{x\in\mathcal{X}_y} \alpha_{xy}(k)\pi_{xy,s} + \frac{\eta}{2} \sum_{x\in\mathcal{X}_y} \left(\pi_{xy}(k) - \pi_{xy,s}\right)^2,
\end{split}
\end{equation}
\begin{equation}\label{ADMM2_eqn3}
\begin{split}
    \pi_{xy}(k+1) = \frac{1}{2} \left(\pi_{xy,t}(k+1) + \pi_{xy,s}(k+1)\right),
\end{split}
\end{equation}
\begin{equation}\label{ADMM2_eqn4}
\begin{split}
    \alpha_{xy}(k+1) = \alpha_{xy}(k) + \frac{\eta}{2}\left(\pi_{xy,t}(k+1) - \pi_{xy,s}(k+1)\right).
\end{split}
\end{equation}
\end{proposition}

\begin{proof} 
The simplification can be obtained straightforwardly by first characterizing the solution to \eqref{ADMM1_eqn3} and then substituting it into \eqref{ADMM1_eqn4} and \eqref{ADMM1_eqn5}.
%As \eqref{ADMM1_eqn3} is strictly concave, we can solve it by first-order condition:
%$
%    \pi_{xy}(k+1) = \frac{1}{2\eta}(\alpha_{xy,t}(k) - \alpha_{xy,s}(k)) + \frac{1}{2}(\pi_{xy,t}(k+1) + \pi_{xy,s}(k+1)).
%$
%By substituting the above equation into  \eqref{ADMM1_eqn4} and \eqref{ADMM1_eqn5} we get:
%$
%    \alpha_{xy,t}(k+1) = \frac{1}{2}(\alpha_{xy,t}(k) + \alpha_{xy,s}(k)) + \frac{\eta}{2}(\pi_{xy,t}(k+1) - \pi_{xy,s}(k+1)),
%$
%$
%   \alpha_{xy,s}(k+1) = \frac{1}{2}(\alpha_{xy,t}(k) + \alpha_{xy,s}(k)) + \frac{\eta}{2}(\pi_{xy,t}(k+1) - \pi_{xy,s}(k+1)).
%$
%We can see that $\alpha_{xy,t} = \alpha_{xy,s}$ during each update. Hence, $\pi_{xy}(k+1)$ can be further simplified as $\pi_{xy}(k+1) = \frac{1}{2} (\pi_{xy,t}(k+1) + \pi_{xy,s}(k+1))$. In addition, we can achieve \eqref{ADMM1_eqn4} and \eqref{ADMM1_eqn5} from  $\alpha_{xy,t} = \alpha_{xy,s} = \alpha_{xy}$ represented in \eqref{ADMM2_eqn4}.
\end{proof}

For convenience, we summarize the distributed OT algorithm into Algorithm \ref{Alg:1}.

\section{Differentially Private Distributed Optimal Transport}\label{sec:DPOT}

In this section, we first present the privacy concerns in the developed distributed OT in Section \ref{sec:DDOT}. We then develop a differentially private distributed OT algorithm that preserves nodes' privacy explicitly during decision updates. 

\subsection{Privacy Concerns in the Distributed OT}
In the previous distributed OT algorithm, the intermediate results are shared between connected nodes during updates. This sharing mechanism raises privacy concerns as an adversary that can access this result (e.g., through eavesdropping attack) has the ability to infer the participants' private information. Specifically, the adversary could leverage the compromised information $\Pi_{x,t}(k)$ and $\Pi_{y,s}(k)$ at each update step, $k$, to infer the node's private information including the sensitive preference parameters in the utility functions $t_{xy}$ and $s_{xy}$. We denote the set of private preference information at node $p$ by $D_{p}$, $p\in\mathcal{P}$.

\begin{algorithm}[!h]
\caption{Distributed OT Algorithm}\label{Alg:1}
\begin{algorithmic}[1]
\While {$\Pi_{x,t}$ and $\Pi_{y,s}$ not converging}
\State Compute $\Pi_{x,t}(k+1)$  using \eqref{ADMM2_eqn1}, for all $x\in\mathcal{X}_y$
\State Compute $\Pi_{y,s}(k+1)$  using \eqref{ADMM2_eqn2}, for all $y\in\mathcal{Y}_x$
\State Compute $\pi_{xy}(k+1)$  using \eqref{ADMM2_eqn3}, for all $\{x,y\}\in \mathcal{E}$
\State Compute $\alpha_{xy}(k+1)$  using \eqref{ADMM2_eqn4}, for all $\{x,y\}\in \mathcal{E}$
\EndWhile
\State \textbf{return} $\pi_{xy}(k+1)$, for all $\{x,y\}\in \mathcal{E}$
\end{algorithmic}
\end{algorithm}

We next use an example to further illustrate node's private information set. Specifically, we consider utility functions admitting a linear form for both the sender and receiver: $t_{xy}(\pi_{xy}) = \delta_{xy}\pi_{xy}$ and $s_{xy}(\pi_{xy}) = \gamma_{xy}\pi_{xy}$, where $\delta_{xy},\gamma_{xy}\in\mathbb{R}_+$. Then, for a target node $x\in\mathcal{X}$, we have set $D_x = \{\delta_{xy}:\forall y\in\mathcal{Y}_x\}$. Similarly, for a source node $y\in\mathcal{Y}$, we have set $D_y = \{\gamma_{xy}:\forall x\in\mathcal{X}_y\}$. The information contained in $D_p$ is crucial for developing optimal transport plans. Leakage of such private information is undesired in many resource allocation scenarios, especially those with societal impacts. For example, in the distribution of scarce vaccine resources, these preference parameters could indicate the severity of epidemics in different neighborhoods (modeled by nodes). It is obvious that each participant does not want to leak this piece of information to other unauthorized parties.

To this end, we aim to protect the privacy of each node in the transport network using differential privacy \cite{dwork2014_book}. Specifically, we propose to add randomness to the transport decisions communicated between each pair of source-target nodes during updates, preventing the adversary from learning the sensitive utility parameters of nodes simply based on the transport decisions. To achieve this goal, first, let $D_{p}$ and $D_{p}'$ be two information/data sets differ by one data point (utility parameter). In other words, their \textit{Hamming Distance} is equal to 1, denoted by $H(D_{p},D_{p}') = 1$. Here, $H(D_{p},D_{p}')= \sum_{i=1}^{|D_p|}\textbf{1}\{i : d_i \neq d_i'\}$, where $d_i$ and $d_i'$ denote the $i^{\mathrm{th}}$ data point in the information set $D_p$ and $D_p'$, respectively. Recall that the data points in these sets refer to the nodes' utility parameters which we aim to protect from leakage under the condition that the adversary intercepts the transport plans. The formal definition of differential privacy is presented below.

\begin{definition}[$\beta_p(k)$-Differential Privacy]\label{beta-DP:def}
Consider the transport network $\mathcal{G}=\{\mathcal{P},\mathcal{E}\}$, where $\mathcal{P}$ is composed of both source nodes and target nodes, and $\mathcal{E}$ is a set of edges connecting the nodes. At each node $p\in\mathcal{P}$, there is an information set $D_{p}$ which is used to compute the resource transport plan. 
%Additionally, $\bar{D} = \cup_{p \in \mathcal{P}} D_{p}$ is the total set of all the information from all the nodes. 
Let $R$ be a randomized counterpart of Algorithm \ref{Alg:1}. Further, let $\beta(k) = \left(\beta_1(k), \beta_2(k),...,\beta_{|\mathcal{P}|}(k)\right) \in \mathbb{R}_+^{|\mathcal{P}|}$, where $\beta_{p}(k) \in \mathbb{R}_+$ is the privacy parameter of node $p$ at iteration $k$. Consider the outputs $\Pi_{x,t}(k)$ and $\Pi_{y,s}(k)$ at iteration $k$ of Algorithm \ref{Alg:1}. Let $D_p'$ be any information set such that $H(D_p',D_p) = 1$ and $\widetilde{\Pi}_x^t(k)$ and $\widetilde{\Pi}_y^s(k)$ be the corresponding outputs of Algorithm \ref{Alg:1} while using the information set $D'_p$. The algorithm $R$ is $\beta_p(k)$-differentially private for any $D'_p$ for all nodes $p \in \mathcal{P}$ and for all possible sets of outcome solutions $S$, if the following condition is satisfied at every iteration $k$:
\begin{equation}
\begin{aligned}
\mathrm{Pr}[\Pi_p(k) \in S] \leq \exp{(\beta_p(k))} \cdot \mathrm{Pr}[\widetilde{\Pi}_p \in S],
\end{aligned}
\end{equation}
where $\Pi_p(k)=\begin{cases}\Pi_{p,t}(k),\ \mathrm{if}\ p\in\mathcal{X},\\ \Pi_{p,s}(k),\ \mathrm{if}\ p\in\mathcal{Y}, \end{cases}$ and $\widetilde{\Pi}_p(k)=\begin{cases}\widetilde{\Pi}_{p,t}(k),\ \mathrm{if}\ p\in\mathcal{X},\\ \widetilde{\Pi}_{p,s}(k),\ \mathrm{if}\ p\in\mathcal{Y}. \end{cases}$
\end{definition}

\subsection{Output Variable Perturbation}
In order to ensure that the sensitive preference information at each node remains private when transport plans are published over the network, we develop a differentially private algorithm based on output variable perturbation.  This algorithm involves adding random noise to the output decision variables $\Pi_{x,t}(k+1)$ and $\Pi_{y,s}(k+1)$ during updates. More specifically, the random noise vectors, $\epsilon_{x}(k+1) \in \mathbb{R}^{|\mathcal{Y}_x|}$ and $\epsilon_{y}(k+1) \in \mathbb{R}^{|\mathcal{X}_y|}$ are added to the variables $\Pi_{x,t}(k+1)$ and $\Pi_{y,s}(k+1)$ obtained by \eqref{ADMM2_eqn1} and \eqref{ADMM2_eqn2}, respectively.

Recall that $p\in\mathcal{P} = \mathcal{X} \cup \mathcal{Y}$ and thus $p = x$, $\forall x \in \mathcal{X}$, and $p = y$, $\forall y \in \mathcal{Y}$. The random noise vector $\epsilon_p(k)$ is generated according to a distribution with density proportional to $e^{-\xi_p(k)||\epsilon||}$. Here, $\xi_p(k) = \frac{\rho}{\eta}\beta_p(k)$, where $\beta_p$ is a privacy term at each node $p$. 
%We test the validity of the density function by confirming that $\int_{0}^{\infty}e^{-\xi_p(k)||\epsilon||} d\epsilon = 1$.
Thus, the proposed solutions at  target node $x$ and source node $y$ at step $k+1$ admit
\begin{equation}
\begin{split}
    \Pi_{x,t}^{*}(k+1) = \Pi_{x,t}(k+1) + \epsilon_{x}(k+1),\\
    \Pi_{y,s}^{*}(k+1) = \Pi_{y,s}(k+1) + \epsilon_{y}(k+1),
\end{split}
\end{equation}
where $\Pi_{x,t}^{*}$ and $\Pi_{y,s}^{*}$ are perturbed solutions of $\Pi_{x,t}$ and $\Pi_{y,s}$, respectively. The distributed OT algorithm with output perturbation includes the following steps:
\begin{equation}\label{ADMM3_eqn1}
\begin{split}
    \Pi_{x,t}(k+1) \in \arg \min_{\Pi_{x}^t\in\mathcal{F}_{x}^t} - \sum_{y\in\mathcal{Y}_x} t_{xy}(\pi_{xy,t}) \qquad \qquad \\
    + \sum_{y\in\mathcal{Y}_x} \alpha_{xy}(k) \pi_{xy,t} + \frac{\eta}{2} \sum_{y\in\mathcal{Y}_x} \left(\pi_{xy,t} - \pi_{xy}(k)\right)^2,
\end{split}
\end{equation}
\begin{equation}\label{ADMM3_pvp1}
\begin{split}
    \Pi_{x,t}^{*}(k+1) = \Pi_{x,t}(k+1) + \epsilon_{x}(k+1),
\end{split}
\end{equation}
\begin{equation}\label{ADMM3_eqn2}
\begin{aligned}
    \Pi_{y,s}(k+1) \in \arg \min_{\Pi_{y,s}\in\mathcal{F}_{y,s}} - \sum_{x\in\mathcal{X}_y} s_{xy}(\pi_{xy,s}) \qquad \qquad \\ -\sum_{x\in\mathcal{X}_y} \alpha_{xy}(k)\pi_{xy,s} + \frac{\eta}{2} \sum_{x\in\mathcal{X}_y} \left(\pi_{xy}(k) - \pi_{xy,s}\right)^2,
\end{aligned}
\end{equation}
\begin{equation}\label{ADMM3_pvp2}
\begin{split}
    \Pi_{y,s}^{*}(k+1) = \Pi_{y,s}(k+1) + \epsilon_{y}(k+1),
\end{split}
\end{equation}
\begin{equation}\label{ADMM3_eqn3}
\begin{split}
    \pi_{xy}^*(k+1) = \frac{1}{2} \left(\pi_{xy,t}^{*}(k+1) + \pi_{xy,s}^{*}(k+1)\right),
\end{split}
\end{equation}
\begin{equation}\label{ADMM3_eqn4}
\begin{split}
    \alpha_{xy}(k+1) = \alpha_{xy}(k) + \frac{\eta}{2}\left(\pi_{xy,t}^{*}(k+1) - \pi_{xy,s}^{*}(k+1)\right).
\end{split}
\end{equation}

\begin{figure}[!t]
    \centering
    \includegraphics[width=.95\columnwidth]{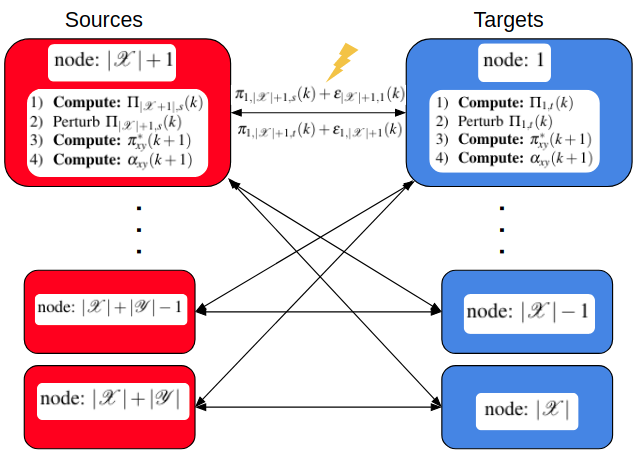}
    \caption{Illustration of the differentially private distributed OT scheme. The information exchanged between nodes is susceptible to be intercepted by the adversary (e.g., by eavesdropping attack to the wireless channel). Hence, an appropriate random noise is added to the outputs at each update step. }
    \label{fig:info-flow}
\end{figure}

As a result of the perturbation in \eqref{ADMM3_pvp1} and \eqref{ADMM3_pvp2}, $\Pi_{x,t}^{*}(k)$ and $\Pi_{y,s}^{*}(k)$ are randomized.  
Specifically, within each iteration, the node perturbs the output variable $\Pi_{x,t}(k)$ or $\Pi_{y,s}(k)$ respectively in order to obtain $\Pi_{x,t}^{*}(k)$ or $\Pi_{y,s}^{*}(k)$. The proposed scheme is further illustrated in Fig. \ref{fig:info-flow}. It is important to note that the information sets at each node, i.e., $D_p$ containing sensitive utility parameters, remain untouched and not perturbed. Due to the random output perturbation, the transport strategy does not converge to a deterministic value compared with the distributed algorithm in Section \ref{sec:dis_alg}. Instead, the algorithm converges approximately and oscillates within a bounded interval. The magnitude of the oscillation is directly related to the differential privacy parameter $\beta_p$ chosen by each node $p\in\mathcal{P}$. When $\beta_p$ becomes larger, $\forall p\in\mathcal{P}$, the differentially privacy algorithm tends to converge to the same solution yielded by Algorithm \ref{Alg:1}.
Since noise is added to each output, the solution will oscillate around the optimal solution. We will test the convergence of the proposed algorithm using case studies. 
%\textcolor{blue}{(We can guarantee the convergence properties of the algorithm because of the known convergence of ADMM \cite{boyd2011distributed}. 
%Therefore, in order to check convergence, we check that the solutions of \eqref{ADMM3_pvp1} and \eqref{ADMM3_pvp2} are within $\max(e^{-\xi(k)||\epsilon||})$.} 
%\textcolor{red}{We did not prove the convergence of the algorithm. How we obtain the bound is $\max(e^{-\xi(k)||\epsilon||})$? Note that this term is just a density function rather than the value that is sampled.)}
For convenience, the differentially private distributed OT algorithm based on the output variable perturbation is summarized in Algorithm \ref{Alg:pvp}.

\begin{algorithm}[!b]
\caption{Differentially Private Distributed OT Algorithm With Output Variable Perturbation}\label{Alg:pvp}
\begin{algorithmic}[1]
\For {$k=0,1,2,...$}
\For {$x\in\mathcal{X}_y$}
\State Compute $\Pi_{x,t}(k+1)$  using \eqref{ADMM3_eqn1}
\State Compute $\Pi_{x,t}^{*}(k+1)$ using \eqref{ADMM3_pvp1}
\EndFor
\For{$y\in\mathcal{Y}_x$}
\State Compute $\Pi_{y,s}(k+1)$  using \eqref{ADMM3_eqn2}
\State Compute $\Pi_{y,s}^{*}(k+1)$ using \eqref{ADMM3_pvp2}
\EndFor
\State Compute $\pi_{xy}^*(k+1)$  using \eqref{ADMM3_eqn3}, for all $\{x,y\}\in \mathcal{E}$
\State Compute $\alpha_{xy}(k+1)$  using \eqref{ADMM3_eqn4}, for all $\{x,y\}\in \mathcal{E}$
\EndFor
\State \textbf{return} $\pi_{xy}^*(k+1)$, for all $\{x,y\}\in \mathcal{E}$
\end{algorithmic}
\end{algorithm}

We further have the following Theorem \ref{thm:pvp} to theoretically guarantee the privacy-preserving property of Algorithm \ref{Alg:pvp}.

\begin{theorem}\label{thm:pvp}
The proposed Algorithm \ref{Alg:pvp} is $\beta$-differentially private with $\beta_p(k)$ for node $p$ at iteration $k$. Let $Q(\Pi_{x,t}^{*}|D_x)$ and $Q(\Pi_{x,t}^{*}|D_x')$ be the probability density functions for $\Pi_{x,t}^{*}$ given the information sets $D_x$ and $D_x'$ such that $H(D_x,D_x') = 1$. The ratio of probability density of $\Pi_{x,t}^{*}$ is bounded:
\begin{equation}\label{pix-bound}
    \frac{Q(\Pi_{x,t}^{*}(k)|D_x)}{Q(\Pi_{x,t}^{*}(k)|D_x')} \leq e^{\beta_x(k)}.
\end{equation}
It follows similarly for the probability density on the source side, $\Pi_{y,s}^{*}$, i.e.,
\begin{equation}\label{piy-bound}
    \frac{Q(\Pi_{y,s}^{*}(k)|D_y)}{Q(\Pi_{y,s}^{*}(k)|D_y')} \leq e^{\beta_y(k)}.
\end{equation}
Note that \eqref{pix-bound} and \eqref{piy-bound} directly imply $\frac{\mathrm{Pr}(\Pi_{x,t}^{*}(k)|D_x)}{\mathrm{Pr}(\Pi_{x,t}^{*}(k)|D_x')} \leq e^{\beta_x(k)}$ and $ \frac{\mathrm{Pr}(\Pi_{y,s}^{*}(k)|D_y)}{\mathrm{Pr}(\Pi_{y,s}^{*}(k)|D_y')} \leq e^{\beta_y(k)}$, respectively.
\end{theorem}

\begin{proof}
We first show the bounded ratio in \eqref{pix-bound}.
We have $\frac{Q(\Pi_{x,t}^{*}(k)|D_x)}{Q(\Pi_{x,t}^{*}(k)|D_x')} =\frac{F_x(\epsilon_x(k))}{F_x(\epsilon_x'(k))}= \frac{e^{-\xi_x(k)||\epsilon_x(k)||}}{e^{-\xi_x(k)||\epsilon_x'(k)||}}$.  Our goal is to find a $\xi_x(k)$ such that the following inequality holds $\xi_x(k)(||\epsilon_x(k)|| - ||\epsilon_x'(k)||) \leq \beta_p(k)$.
Let $W = \arg \min_{\Pi_{x,t}} f_x(k|D_x)$ and $W' = \arg \min_{\Pi_{x,t}}f_x(k|D_x')$, where $f_x(k)$ is the objective function for the target node $x\in\mathcal{X}$ at iteration $k$, shown in \eqref{ADMM3_eqn1}. Also, let $g$ and $h$ be defined at each node $x \in \mathcal{X}$ such that $g(\Pi_{x,t}^{*}(k)) = f_x(k|D_x)$ and $h(\Pi_{x,t}^{*}(k)) = f_x(k|D_x') - f_x(k|D_x)$.

Therefore,
$
h(\Pi_{x,t}^{*}(k)) = - \tilde{t}_{xy}(\pi_{xy,t}) + t_{xy}(\pi_{xy,t}),
$
where $\tilde{t}_{xy}$ refers to the altered utility function due to the difference between $D_x'$ and $D_x$.
Assumption \ref{assump:1} implies that $f_x(k|D_p) = g(\Pi_{x,t}^{*}(k))$ and $f_x(k|D_x') = g(\Pi_{x,t}^{*}(k)) + h(\Pi_{x,t}^{*}(k)) $ are both convex. We differentiate $h(\Pi_{x,t}^{*}(k))$ with respect to $\Pi_{x,t}^{*}(k)$ and get:
%\begin{equation*}
%\begin{split}
%    \nabla h(\Pi_{x,t}^{*}(k)) =\qquad\qquad\qquad \qquad\qquad\qquad\qquad\qquad  \\ \pi_{xy,t}' \left(t_{xy}(\pi_{xy,t}') + \alpha_{xy}(k)\pi_{xy,t}' +\frac{\eta}{2}(\pi_{xy,t}' - \pi_{xy}(k))^2 \right)\\
%    +\pi_{xy,t} \left( t_{xy}(\pi_{xy,t}) + \alpha_{xy}(k)\pi_{xy,t} +\frac{\eta}{2}(\pi_{xy,t} - \pi_{xy}(k))^2  \right)
%\end{split}
%\end{equation*}
\begin{equation*}
    \nabla h(\Pi_{x,t}^{*}(k)) = -\tilde{t}'_{xy}(\pi_{xy,t}) + t'_{xy}(\pi_{xy,t}).
\end{equation*}
Assumption \ref{assump:1} further implies that $0 \leq t'_{xy} \leq \rho$. Thus, $||\nabla h(\Pi_{x,t}^{*})|| \leq \rho$. From the definitions of $W$ and $W'$, we have $\nabla g(W) = \nabla g(W') + \nabla h(W') = 0$. Based on Lemma 14 in \cite{Schwartz2007} and knowing that $g(\cdot)$ is $\eta$-strongly convex, the following inequality holds: $\langle \nabla g(W) - g(W'), W-W'\rangle \geq  \eta || W - W' ||^2.$ Thus, by Cauchy-Schwartz inequality, we obtain
\begin{equation*}
\begin{split}
    || W - W' || \cdot ||\nabla h(W')|| \geq (W -W')^T \nabla h(W') = \\
    \langle \nabla g(W) - g(W'), W - W' \rangle \geq \eta ||W - W'||^2.
\end{split}
\end{equation*}
Dividing both sides by $\eta ||W-W'||$ yields 
$
    ||W-W'|| \leq \frac{1}{\eta}|| \nabla h(W')|| \leq \frac{\rho}{\eta}.
$
From \eqref{ADMM3_pvp1}, we have
$
    ||W-W'|| = ||\epsilon_x(k) - \epsilon'_x(k)|| \leq \frac{1}{\eta}||\nabla h(W')||.
$
Thus, we obtain
$$
    \xi_x(k) (||\epsilon_x(k)|| - ||\epsilon'_x(k)||) \leq \xi_x(k)(||\epsilon_x(k) - \epsilon'_x(k)||) \leq \frac{\rho}{\eta}\xi_x(k).
$$
By choosing $\xi_x(k) = \frac{\eta}{\rho}\beta_p(k)$, the inequality $\xi_x(k) (||\epsilon_x(k) - \epsilon'_x(k)||) \leq \beta_p(k)$ holds. Thus, the output variable perturbation is $\beta_p$-differentially private for target node $x \in \mathcal{X}$. The proof follows identically for the perturbed output variable $\Pi_{y,s}^{*}(k)$ at the source node $y \in \mathcal{Y}$ and hence omitted.
\end{proof}

In summary, the proposed Algorithm \ref{Alg:pvp} guarantees the privacy of all participating nodes during their decision sharing. 

\section{Numerical Case Studies}\label{sec:case}
In this section, we corroborate the effectiveness of the developed differentially private algorithm and show how the added privacy impacts the transport plan and its efficiency.

We construct a transport network with four source nodes and thirty target nodes in which every source node is connected to all target nodes, i.e., the network is complete. The upper bounds at the target nodes $\bar{p}_x$ are kept small (smaller than 5), while the upper bounds at the source nodes $\bar{q}_y$ are relatively larger (between 20 and 40). Such selection yields that the resources at the origin can be transported to heterogeneous target nodes. Additionally, we consider linear utility functions $t_{xy}(\pi_{xy}) = \delta_{xy}\pi_{xy}$, and $s_{xy}(\pi_{xy}) = \gamma_{xy}\pi_{xy}, \forall\{x,y\} \in \mathcal{E}$. The utility parameters $\delta_{xy}$ and $\gamma_{xy}$ are randomly chosen integers between 1 and 5 for each pair of connection, $\forall \{x,y\} \in \mathcal{E}$. 
%Thus, we get two $30 \times 4$ matrices for $t_{xy}$ and $s_{xy}$ $\forall \{x,y\} \in \mathcal{E}$.

In the following study, we investigate the impact of privacy parameter $\beta_p$ on the transport utility. According to the definition, a smaller $\beta_p$ yields a higher level of privacy. We compare the results for two sets of $\beta_p$. For the first one, we assign a value of $1$ to $\beta_p$, $p\in\mathcal{P}$. For the larger value of $\beta_p$  we use 1000. Furthermore, we select $\eta =1$ and $\rho=2$.

\begin{figure*}[!t] 
\centering
\subfigure[Social Utility]{\includegraphics[width=0.69\columnwidth]{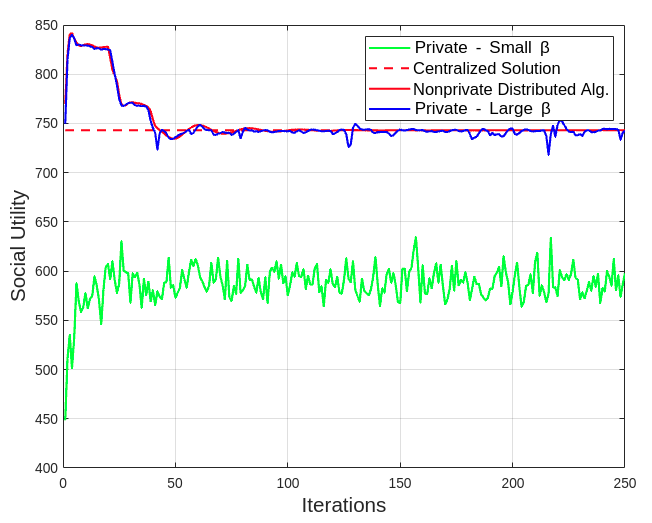}\label{fig:f2_1}}
\subfigure[Transport Plan]{\includegraphics[width=0.67\columnwidth]{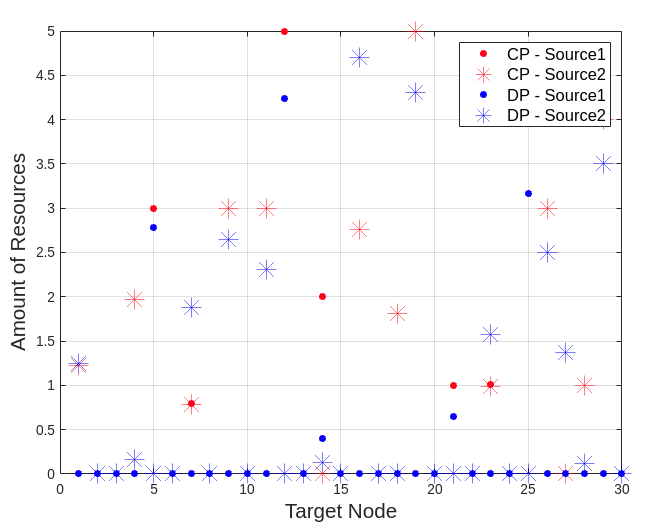}\label{fig:f2_2}}
\subfigure[Privacy and Transport Efficiency Tradeoff ]{\includegraphics[width=0.66\columnwidth]{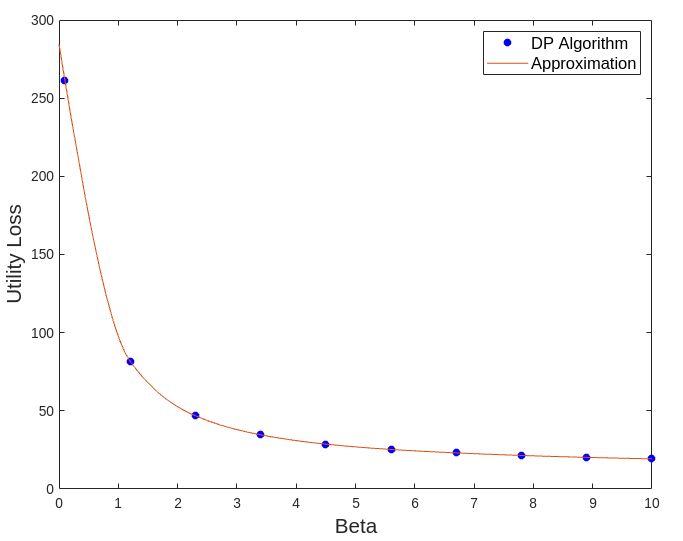}\label{fig:f2_3}}
\caption{(a) shows the performance of the proposed algorithms. (b) depicts the optimal transport plans designed by the central planner (CP) and the solution given by the distributed differentially private (DP) algorithm. (c) shows an increase of the privacy level (smaller $\beta_p$) decreases the transport utility, reflecting the trade-off between privacy and transport efficiency.}
\label{fig:f2}
\end{figure*}

We leverage Algorithms \ref{Alg:1} and \ref{Alg:pvp} to compute the transport plans. The results are shown in Fig. \ref{fig:f2}. First, we observe that in Fig. \ref{fig:f2_1}, the trajectory of transport plan yielded by the differentially private algorithm converges approximately to a certain value. The oscillation at the tail is due to the random noise added to the decision at each output perturbation step. We can also see that when $\beta_p$ is small, the resulting social utility (i.e., transport efficiency), which is an aggregation of the utilities of all participating nodes, is relatively small.  In comparison, when $\beta_p$ is large, the social utility is close to the one returned by Algorithm \ref{Alg:1} where differential privacy is not incorporated. Fig. \ref{fig:f2_3} further shows this phenomenon and reveals the inherent trade-off between the amount of added privacy and the transport efficiency. Fig. \ref{fig:f2_2} illustrates how the privacy factor affects the transport plan. The decreased optimality due to the privacy promotion indicates that the resource allocation is no longer taking full advantage of how much source nodes can provide or how much target nodes can request. For example, the target node 12 can request at most 5 units of resources, and does so when privacy is not added to the algorithm. When privacy is concerned, it only requests and receives 4.2 units of resources and hence the social utility is decreased. 
%The decreased optimality comes from the node not getting the greatest amount of resources it can.

\section{Conclusion}\label{sec:conclusion}
This paper has developed a differentially private distributed optimal transport algorithm with a theoretical guarantee of achieved privacy. The algorithm protects the sensitive information at each node by perturbing the output of the transport schemes shared between connected nodes during updates. Under the designed mechanism, even if the transport decision is intercepted during its transmission, the adversary still cannot discover the underlying sensitive information used in the transport strategy design. The privacy level for each node can be determined appropriately by considering its trade-off with the resulting transport efficiency. Future work includes extending the current model-based distributed optimal transport framework to data-driven learning-based optimal transport while considering data privacy in the learning process.

\bibliographystyle{IEEEtran}
\bibliography{references.bib}

\begin{thebibliography}{10}
\providecommand{\url}[1]{#1}
\csname url@samestyle\endcsname
\providecommand{\newblock}{\relax}
\providecommand{\bibinfo}[2]{#2}
\providecommand{\BIBentrySTDinterwordspacing}{\spaceskip=0pt\relax}
\providecommand{\BIBentryALTinterwordstretchfactor}{4}
\providecommand{\BIBentryALTinterwordspacing}{\spaceskip=\fontdimen2\font plus
\BIBentryALTinterwordstretchfactor\fontdimen3\font minus
  \fontdimen4\font\relax}
\providecommand{\BIBforeignlanguage}[2]{{%
\expandafter\ifx\csname l@#1\endcsname\relax
\typeout{** WARNING: IEEEtran.bst: No hyphenation pattern has been}%
\typeout{** loaded for the language `#1'. Using the pattern for}%
\typeout{** the default language instead.}%
\else
\language=\csname l@#1\endcsname
\fi
#2}}
\providecommand{\BIBdecl}{\relax}
\BIBdecl

\bibitem{jhughes2021fair}
J.~Hughes and J.~Chen, ``Fair and distributed dynamic optimal transport for
  resource allocation over networks,'' in \emph{55th Annual Conference on
  Information Sciences and Systems (CISS)}, 2021.

\bibitem{zhang2019consensus}
R.~Zhang and Q.~Zhu, ``Consensus-based distributed discrete optimal transport
  for decentralized resource matching,'' \emph{IEEE Transactions on Signal and
  Information Processing over Networks}, vol.~5, no.~3, pp. 511--524, 2019.

\bibitem{hughes2021deceptive}
J.~Hughes and J.~Chen, ``Resilient and distributed discrete optimal transport
  with deceptive adversary: A game-theoretic approach,'' in \emph{IEEE Control
  System Letters}, 2022, pp. 1166--1171.

\bibitem{dwork2014_book}
C.~Dwork, A.~Roth \emph{et~al.}, ``The algorithmic foundations of differential
  privacy.'' \emph{Foundations and Trends in Theoretical Computer Science},
  vol.~9, no. 3-4, pp. 211--407, 2014.

\bibitem{2020HuangDPADMM}
Z.~Huang, R.~Hu, Y.~Guo, E.~Chan-Tin, and Y.~Gong, ``{DP-ADMM}: {ADMM}-based
  distributed learning with differential privacy,'' \emph{IEEE Transactions on
  Information Forensics and Security}, vol.~15, p. 1002–1012, 2020.

\bibitem{czhang2019ADMMopt}
C.~Zhang, M.~Ahmad, and Y.~Wang, ``{ADMM} based privacy-preserving
  decentralized optimization,'' \emph{IEEE Transactions on Information
  Forensics and Security}, vol.~14, no.~3, pp. 565--580, 2019.

\bibitem{zhang2017privateADMMlearning}
T.~Zhang and Q.~Zhu, ``Dynamic differential privacy for {ADMM}-based
  distributed classification learning,'' \emph{IEEE Transactions on Information
  Forensics and Security}, vol.~12, no.~1, pp. 172--187, 2017.

\bibitem{zhang2018improving}
X.~Zhang, M.~M. Khalili, and M.~Liu, ``Improving the privacy and accuracy of
  {ADMM}-based distributed algorithms,'' in \emph{International Conference on
  Machine Learning}.\hskip 1em plus 0.5em minus 0.4em\relax PMLR, 2018, pp.
  5796--5805.

\bibitem{chaudhuri2011privateERM}
K.~Chaudhuri, C.~Monteleoni, and A.~D. Sarwate, ``Differentially private
  empirical risk minimization,'' \emph{Journal of Machine Learning Research},
  vol.~12, no.~29, pp. 1069--1109, 2011.

\bibitem{zhang2019privateSVM}
Y.~Zhang, Z.~Hao, and S.~Wang, ``A differential privacy support vector machine
  classifier based on dual variable perturbation,'' \emph{IEEE Access}, vol.~7,
  pp. 98\,238--98\,251, 2019.

\bibitem{Abadi2016privateDL}
M.~Abadi, A.~Chu, I.~Goodfellow, H.~B. McMahan, I.~Mironov, K.~Talwar, and
  L.~Zhang, ``Deep learning with differential privacy,'' \emph{Proceedings of
  the ACM SIGSAC Conference on Computer and Communications Security}, 2016.

\bibitem{du2017differential}
M.~Du, K.~Wang, X.~Liu, S.~Guo, and Y.~Zhang, ``A differential privacy-based
  query model for sustainable fog data centers,'' \emph{IEEE Transactions on
  Sustainable Computing}, vol.~4, no.~2, pp. 145--155, 2017.

\bibitem{khalili2021designing}
M.~M. Khalili, X.~Zhang, and M.~Liu, ``Designing contracts for trading private
  and heterogeneous data using a biased differentially private algorithm,''
  \emph{IEEE Access}, vol.~9, pp. 70\,732--70\,745, 2021.

\bibitem{zhao2020local}
Y.~Zhao, J.~Zhao, M.~Yang, T.~Wang, N.~Wang, L.~Lyu, D.~Niyato, and K.-Y. Lam,
  ``Local differential privacy-based federated learning for {I}nternet of
  things,'' \emph{IEEE Internet of Things Journal}, vol.~8, no.~11, pp.
  8836--8853, 2020.

\bibitem{zhang2018distributed}
T.~Zhang and Q.~Zhu, ``Distributed privacy-preserving collaborative intrusion
  detection systems for {VANETs},'' \emph{IEEE Transactions on Signal and
  Information Processing over Networks}, vol.~4, no.~1, pp. 148--161, 2018.

\bibitem{boyd2011distributed}
S.~Boyd, N.~Parikh, and E.~Chu, \emph{Distributed Optimization and Statistical
  Learning via the Alternating Direction Method of Multipliers}.\hskip 1em plus
  0.5em minus 0.4em\relax Now Publishers, 2011.

\bibitem{Schwartz2007}
S.~Shalev-Shwartz, ``Online learning: Theory, algorithms, and applications,''
  \emph{PhD Dissertation}, 2007.

\end{thebibliography}

\appendix

\subsection{Proof of Proposition \ref{prop:1}}\label{proof_prop:1}
\begin{proof}
Let $\Vec{x} = [\Vec{\Pi}_{x,t}^T, \Vec{\Pi}^T]^T$, $\Vec{y} = [\Vec{\Pi}^T, \Vec{\Pi}_{y,s}^T]^T$, and $\alpha = [\{\alpha_{xy,t}\}^T, \{\alpha_{xy,s}\}^T]^T$, where $\Vec{}$ denotes the vectorization operator. We note that these vectors are all $2|\mathcal{E}| \times 1$, where $|\mathcal{E}|$ denotes the number of connections between targets and sources. Now we can write the constraints in matrix form such that $A\Vec{x} = \vec{y}$ where $A = [\textbf{I},\textbf{0},\textbf{I},\textbf{0}]$. Here $\textbf{I}$ and $\textbf{0}$ denote the identity and zero matrices respectively, both of which are  $|\mathcal{E}| \times |\mathcal{E}|$. Next, we note that $\Vec{x} \in \mathcal{F}_{\Vec{x},t}$ and $\Vec{y} \in \mathcal{F}_{\Vec{y},s}$, where
$
      \mathcal{F}_{\Vec{x},t} = \{ \Vec{x} | \pi_{xy,t} \geq 0, \underline{p}_x \leq \sum_{y \in \mathcal{Y}_x} \pi_{xy,t} \leq \bar{p}_x, \{x,y\} \in \mathcal{E} \},\ 
       \mathcal{F}_{\Vec{y},s} := \{ \Vec{y} | \pi_{xy,s} \geq 0, \underline{q}_y \leq \sum_{x \in \mathcal{X}_y} \pi_{xy,s} \leq \bar{q}_y, \{x,y\} \in \mathcal{E}  \}.
$
In turn we can solve the minimization in \eqref{OT2:eqn} with the iterations: 1)
$
    \Vec{x}(k+1) \in \arg \min_{\Vec{x} \in \mathcal{F}_{x,t}} L(\Vec{x},\Vec{y}(k),\alpha(k));
$
2)
$
    \Vec{y}(k+1) \in \arg \min_{\Vec{y} \in \mathcal{F}_{y,s}} L(\Vec{x}(k),\Vec{y},\alpha(k));
$
3)
$
    \alpha(k+1) = \alpha(k) + \eta(A\Vec{x}(k+1) - \Vec{y}(k+1)),
$
whose convergence is proved \cite{boyd2011distributed}.
Because there is no coupling among $\Pi_{x,t}, \Pi_{y,s}, \pi_{xy}, \alpha_{xy,t},$ and $\alpha_{xy,s}$, the above iterations can be decomposed to \eqref{ADMM1_eqn1}-\eqref{ADMM1_eqn5}. 
\end{proof}

\end{document}